\documentclass{article}
\usepackage[utf8]{inputenc} 

\usepackage{amssymb,amsfonts,amsmath,mathtext,cite,enumerate,float,amsthm}
\usepackage{ytableau}
\usepackage[vcentermath]{youngtab}
\usepackage{hyperref}
\usepackage{tikz}
\usetikzlibrary{positioning}
\usetikzlibrary{arrows}
\graphicspath{ {C:/Users/yourf/OneDrive/Documents/Наука/ТеХ/} }
\usepackage{tikz-cd}
\usepackage[top = 2 cm, bottom = 2 cm, left = 2.5 cm, right = 1.5 cm]{geometry} 

\theoremstyle{definition}
\newtheorem{Definition}{Definition}[section]
\theoremstyle{plain}
\newtheorem{Lemma}{Lemma}[section]
\newtheorem{Conjecture}{Conjecture}[section]
\newtheorem{Statement}{Statement}[section]
\newtheorem{Proposition}{Proposition}[section]
\newtheorem{Theorem}{Theorem}[section]
\newtheorem{Corollary}{Corollary}[section]
\theoremstyle{remark}
\newtheorem*{Remark}{Remark}

\newenvironment{innerproof}
 {\proof}
 {\endproof}

\newcommand{\eqnb}{\begin{equation}}
\newcommand{\eqn}{\end{equation}}

\newcommand{\h}{\hbar}
\newcommand{\hh}{\hat{h}}

\newcommand{\Al}{\mathcal{A}l}
\newcommand{\K}{\mathcal{K}}

\newcommand{\A}{\mathcal{A}}
\newcommand{\fullKP}{\widehat{\mathcal{KP}}}
\newcommand{\KP}{\mathcal{KP}}
\newcommand{\Y}{\mathcal{Y}}

\DeclareMathOperator{\Span}{span}

\begin{document}

\title{\vspace{0.1cm}{\Large {\bf 
Perturbative analysis of the colored  Alexander polynomial and KP soliton $\tau$-functions  }\vspace{.2cm}}
\author{
	{\bf V. Mishnyakov $^{a,c,d}$}\thanks{mishnyakovvv@gmail.com},
	{\bf A. Sleptsov$^{a,b,d}$}\thanks{sleptsov@itep.ru}} \date{ }%
}
\maketitle

\vspace{-5.5cm}

\begin{center}
	\hfill MIPT/TH-07/19\\
	\hfill IITP/TH-08/19\\
	\hfill ITEP/TH-13/19
\end{center}

\vspace{3.3cm}

\begin{center}
	$^a$ {\small {\it ITEP, Moscow 117218, Russia}}\\
	$^b$ {\small {\it Institute for Information Transmission Problems, Moscow 127994, Russia}}\\
	$^c$ {\small {\it Moscow State University, Physical Department,
			Vorobjevy Gory, Moscow, 119899, Russia }}\\
	$^d$ {\small {\it Moscow Institute of Physics and Technology, Dolgoprudny 141701, Russia	 }}
\end{center}

\vspace{1cm}

\begin{abstract}
In this paper we study the group theoretic structures of colored HOMFLY polynomials in a specific limit. The group structures arise in the perturbative expansion of $SU(N)$ Chern-Simons Wilson loops, while the limit is $N \rightarrow 0$. The result of the paper is twofold. First, we explain the emergence of Kadomsev-Petviashvily (KP) $\tau$-functions. This result is an extension of what we did in \cite{MMMMS}, where a symbolic correspondence between KP equations and group factors was established. In this paper we prove that integrability of the colored Alexander polynomial is due to it's relation to soliton $\tau$-functions. Mainly, the colored Alexander polynomial is embedded in the action of the KP generating function on the soliton $\tau$-function. Secondly, we use this correspondence to provide a rather simple combinatoric description of the group factors in term of Young diagrams, which is otherwise described in terms of chord diagrams, where no simple description is known. This is a first step providing an explicit description of the group theoretic data of Wilson loops, which would effectively reduce them to a purely topological quantity, mainly to a collection of Vassiliev invariants.
\\\\
\textsc{MSC2020 Subject Classification}: 57K14, 57K16, 37K10, 81R12.

\end{abstract}

\vspace{.5cm}


\section{Introduction}

Wilson loops are important observables in quantum gauge theory, for example, because of their possible role in 4d confinement \cite{Pol}. In three dimensional topological theories  like the Chern-Simons theory \cite{Witten,ChernSimons,CS} one can thoroughly investigate properties of Wilson loops. In this case they provide knot invariants such as the Jones, HOMFLY or Kauffman polynomials depending on the gauge group \cite{Witten}. Apart from Chern-Simons theory knot invariants appear in various branches of mathematical physics such as quantum groups \cite{RT}, lattice models \cite{StatPhys}, WZW models \cite{WZNW}, topological strings \cite{TopStrings},
quantum computing \cite{MMMMM} etc. and hence attracts a lot of attention nowadays.\\\\
Among the important properties of QFT observables are the relations they satisfy \cite{Morozov:2013xra}. Most commonly studied are the Ward identities, which reflect the independence of the path integral on the integration variables. In this paper we are going to focus on the representation theoretic properties of Wilson loops. Actually, since the Chern-Simons theory is topological, for a given knot the only non-trivial dependence of the observable is the representation. Globally we are aiming for a full description of the group-theoretic dependence in the perturbative expansion of the Wilson loop. It is tempting to find a set of representation theoretic properties which would fix the possible structures appearing in the coefficients of the perturbative expansion. In fact this question is not unique to Chern-Simons theory, dissecting the group theory structures is important for any non-abelian theory \cite{FD}, see for example a recent paper for a somewhat similar discussion of $\mathcal{N}=4$ SYM Wilson loops  \cite{N=4SYM}. The specific case of Chern-Simons theory is distinguished due its topological nature, which makes a general investigation easier and also establishes a connection to various problems in knot theory. Moreover, we are able to use results from the theory of quantum knot invariants, where various methods for exact calculations where developed.
\\\\
The fact that knot ivariants such as the HOMFLY polynomial are \emph{exactly} calculable allows to search for their integrable properties. An obvious relation to quantum integrable systems is provided by the $R$-matrix/braid representation formulation of knot invariants. This connects knot invariants to state-sum and vertex models. We are going to discuss a complementary idea of classical KP integrability, which was hinted in our previous paper \cite{MMMMS}. Remarkably, this emergent KP hierarchy is also related to the special representation dependence. And since the integrable KP hierarchy is closely related to the representation theory of the $gl(\infty)$ algebra, our research allows us to hope to obtain a complete answer for the Wilson loop in terms of integrability.
\\\\
The integrable hierarchy of nonlinear partial differential equations Kadomtsev-Petviashvili (KP hierarchy) arose from the Kadomtsev-Petviashvili equation, which generalizes the Korteweg-de Vries equation to 2 spatial dimensions. Initially, this equation described nonlinear waves (solitons) in various media, for example, on water or in a ferromagnetic medium. However, over the past 30 years, it has been found that many partition functions in theoretical physics also satisfy this hierarchy. Among the well-known examples are Kontsevich matrix model \cite{Kont,KhMMMZ}, which describes 2d topological gravity and coincides with a partition function of physical 2d quantum gravity according to Witten's conjecture \cite{Witgrav}; the matrix Brezin-Gross-Witten model  describing the partition function of lattice QCD with Wilson action \cite{BG,GW}; the partition function of the 2d Yang-Mills theory with the gauge group $U(N)$ \cite{Rusakov}; the Hermitian matrix model, which is dual to the  Jackiw-Teitelboim gravity in a particular regime \cite{JT}, which describes 2d dilaton gravity; conformal blocks of 2d conformal Liouville field theory and Nekrasov instanton partition function by virtue of AGT relation \cite{Mironov:2017lgl}.
\\\\
Another long standing problem in knot theory is the problem of describing the so called Vassiliev invariants \cite{ChmutovDuzhin} (or finite-type invariants). These are rational valued invariants graded by a natural number called their order, which are thought to describe a complete set of knot invariants. They are known to have combinatoric description in terms of chord diagrams and a lot of theorems have been proven in this regard. However not much is known about their particular values. Moreover even the number of independent invariants at a given order is unknown. Various asymptotic bounds are given in mathematical literature. Physically, they appear in the perturbative expansion of Wilson loops describing its particular embedding (up to ambient isotopy) in the given 3-dimensional manifold, and are given by tedious integrals, which are of course impossible to calculate \cite{La,Labast,DBSlSm}. The only practical method of finding Vassiliev invariants is by using exact results for HOMFLY polynomials given by R-matrix calculations. Two problems arise on this way. The obvious one is the complexity of calculation of HOMFLY polynomials for large representations. However, even knowing the answer for a given representation one has to know the independent group-theoretic structures in order to distinguish independent Vassiliev invariants.
\\\\
Our paper is organized as follows.
After a brief review of properties of Wilson loops and knot invariants (sec.\ref{sec:2}) we outline the results of the paper in sec. \ref{sec:3}. There we discuss the qualitative aspects of the statements, their importance and physical consequences.
The rest of the paper is devoted to a more mathematical exposition and proofs of the statements. The colored Alexander polynomial is shortly introduced as a generalization of the usual topologically defined Alexander polynomial in section \ref{sec:4}. In section \ref{sec:5} we study the perturbative expansion of the polynomial and define the Alexander system of equations. We are interested in studying the space of it's solutions. We derive properties of these solutions and the dimension of this space. We also mention an interesting connection to supersymmetric polynomials. Section \ref{sec:6} is devoted to the general properties of the KP equations and $\tau$-functions required for the formulation of our statement. Finally, in section \ref{sec:7} we will prove the main theorem stating the correspondence between the solutions of the Alexander equations and the soliton dispersion relations and the combinatoric description of perturbative group factors.
\section{Knot invariants}\label{sec:2}
As mentioned above HOMFLY polynomials are Wilson loop averages specifically for Chern-Simons theory with $SU(N)$ gauge group defined by the knot and the representation $R$:
\begin{equation}
H_R^\mathcal{K}(q,a)=   \left\langle \mathrm{ tr_R } \ P \mathrm{exp} \left(\oint A_\mu^a (x) T^a dx^\mu \ \right)  \right\rangle
\end{equation}
It is a polynomial of two variables:
$$q=e^\h, \ a=e^{N \h},\ \ \ \h:={2\pi i\over \kappa+N}.$$
with $\kappa$ being the Chern-Simons level. The perturbative expansion is recovered using the above parametrisation and expanding in $\hbar$. There are two natural limits one could take:
\begin{equation}
    q  \text{ - fixed} , \ A = 1 \qquad \text{or } \qquad q=1  , \ A  \text{ - fixed}
\end{equation}
The second limit corresponds to the t'Hooft planar limit $
\h \rightarrow 0 , N \rightarrow \infty$ such that $N\h = const$ and produces the so called special polynomials ${\cal H}^\K_R(1,a)=\sigma_R^\K(a)$. Their $R$-dependence has a simple form  \cite{DMMSS,Itoyama,Zhu,GenusExpansion}:
\begin{equation}\label{special}
    \sigma_R^\K(a)= \left( \sigma_{[1]}^\K(a) \right)^{|R|}.
\end{equation}
Hereafter, we identify the representation $R$ with the Young diagram associated with it: $R=\{R_i\},\ R_1\ge R_2\ge\ldots\ge R_{l(R)}$,
$|R|:=\sum_i R_i$. This limit exhibits KP-integrability \cite{inttorus}, namely one can construct a KP $\tau$-function. Indeed, let us consider the Ooguri-Vafa partition function \cite{GenusHurwitz}, which describes the expansion of the Chern-Simons theory around the trivial flat connection:
\begin{equation}
\label{OVpf}
\mathcal{Z}^{\K}(\bar{p}|a,q) =  \sum_R {\cal H}^{\K}_R(a,q)\,D_R\,\chi_R\{\bar{p}\}.
\end{equation}
In this formula, $\chi_R\{p\}$ is the Schur polynomial in the representation $R$, $D_R=\chi_R\{p^*\}$ with $p_k^*=\frac{a^k-a^{-k}}{q^k-q^{-k}}$ being the quantum dimension. Now we put $q=1$ in this formula and HOMFLY polynomial turns into the special polynomial. It turns out \cite{GenusExpansion} that the coefficients in front of Schur polynomials in formula \eqref{OVpf} satisfy the bilinear Plucker identities \emph{specifically due to their special representation property \eqref{special}}.
\\\\
In this paper we will fully focus on the second limit of the HOMFLY polynomial, which gives the so called Alexander polynomial:
\begin{equation}
    H_{R}^\mathcal{K}(q,1)=\mathcal{A}_{R}^\mathcal{K}(q)
\end{equation}
As outlined in \cite{MMMMS} the second limit may be called dual to the t'Hooft limit as requires $N \rightarrow 0$ and exhibits a "dual" property \cite{Itoyama,MM}\footnote{Let us note that our notion of coloured Alexander polynomial is totally different from that defined in \cite{Mur}.}:
\begin{equation}\label{alexproperty0}
\A^\K_R(q)=\A^\K_{[1]}(q^{\vert R\vert})\ ,\quad \text{for} \ R=[r,1^L].
\end{equation}
An important distinction, however, is that the property \eqref{special} holds for all representations, while \eqref{alexproperty0} only holds for single hook Young diagrams. In quantum field theory language this is the relation between Wilson loops in different representations (though in a specific limit) mentioned above. The goal of this paper is to explain the importance of this relation, the restriction it imposes on the structure of the coefficients of the loop expansion and its relation to KP-integrability. In our previous paper \cite{MMMMS} we have exposed a "phenomenological" correspondence, between these coefficients and equations of the KP hierarchy.
\\\\
We are going to study the structure of the loop expansion of HOMFLY polynomials. In mathematical literature it goes under the name "Kontsevich integral" \cite{Labast,Hidden,Sm}, which schematically can be written as follows:
\begin{equation}\label{Expansion}
{\cal H}_R^\K=\sum_n \left(\sum_j v_{n,j}^\mathcal{K}\, r_{n,j}^R\right) \h^n.
\end{equation}
A remarkable fact is that the knot dependence and the group theoretic dependence split explicitly. The group is represented by the so called group factors $r_{n,j}^R$. They appear in the path-integral as traces of $SU(N)$ generators 
\begin{equation}
r^R_{n,j} \sim \mathrm{tr}_R(T^{a_1}\ldots T^{a_n}).
\end{equation}
The knot dependent part is given by  Vassiliev invaraints $v_{n,j}^\mathcal{K}$, which are rational valued knot invariants \cite{Labast,CDIntro}. 
\\\\ 
The group factors $r^R_{n,j}$ are commonly studied in terms of algebras of Jacobi or chord diagrams with weights associated to them. Physically \cite{Labast} these diagrams are the Feynman diagrams appearing in the perturbative expansion of the Wilson loop. An elaborate discussion is given in \cite{CDIntro} (see sections 5,6). Jacobi diagrams and chord diagrams are algebras under tensor products satifying some complicated relations.
To construct knot invariants one associates Lie algebra weight systems to any Jacobi diagram in a manner which is clear from the following examples:

\begin{picture}(850,80)(-100,-70)
\put(-35,-22){$D_2=$}
\put(35,-22){$, \ \varphi_{\mathfrak{g}}(D_2)=\displaystyle \sum_{a,b,c=1}^{\dim \mathfrak{g}}T^a T^b T^c T^{a*} T^{b*} T^{c*}$}
\put(7,-34){\line(0,1){28}}
\put(23,-34){\line(0,1){28}}
\put(-1,-20){\line(1,0){32}}
\put(15,-20){\circle{30}}
\put(8,-18){\mbox{\fontsize{7.5}{7.5} $a$}}
\put(5,-30){\mbox{\fontsize{7.5}{7.5} $b$}}
\put(15,-14){\mbox{\fontsize{7.5}{7.5} $c$}}
\end{picture}
\begin{picture}(850,30)(-300,-50)
\put(-205,-22){$D_1=$}
\put(-160,-20){\circle{30}}
\put(-176,-20){\line(1,0){32}}
\put(-140,-22){$,\ \varphi_{\mathfrak{g}}(D_1)=\displaystyle \sum_{a=1}^{\dim \mathfrak{g}}T^a T^{a*} ;$}
\end{picture}
where $T^a$ are the generators of the Lie algebra $\mathfrak{g}$. This construction is naturally extended to trivalent diagrams, which are also called Jacobi diagrams. By theorem 6.1.2 in \cite{CDIntro} $\varphi$ gives a homomorphism to the center of the universal enveloping algebra $U(\mathfrak{g})$. The group factors are constructed in the same manner for a general representation by composing the representation with the homomophism $\varphi_{\mathfrak{g}}$ and taking the trace. As they are traces of the representation of the center $ZU(\mathfrak{g})$ they they can be expanded into the basis of the Casimir invariants of the algebra \cite{Labast}. Obviously exactly the same holds for the group factors of the Alexander polynomial. For a chord diagram $D$:
\begin{eqnarray}
r_{D}^R=\sum_{\Delta} \alpha_{\Delta,D} C_\Delta(R)
\end{eqnarray}

Let us stress that the precise dimension of the algebra of chord diagram for arbitrary $n$ is unknown, because the algebra relations become very complicated. The explicit basis is known up to the 9th order, obtained by linear algebra approach in \cite{BarNatan}, whereas the precise dimension is known up to the 12th order \cite{Kneissler}. Thus, in order to calculate, for example, the 10th order, one must first find a basis there (the dimension of the algebra at this level is 27). After that, for each diagram, it will be necessary to calculate its weight function for an arbitrary rank $N$ and an arbitrary finite-dimensional representation corresponding to the Young diagram $R$. This procedure seems possible, but extremely laborious and tedious. However, it is not yet possible to move beyond the 12th order in this way.
In this paper, we propose a completely different approach to calculating weight systems of $sl_N$, based on the symmetries of the Wilson correlators. Instead of a monstrous algebra of chord diagrams with obscure relations, we construct a tame ring of polynomials, which will be enumerated by simple combinatorial data (some particular subset of Young diagrams) and present explicit formulas for their calculation.


\section{Summary of results}\label{sec:3}
Let us first state and discuss the results of the paper omitting all details, which are presented in the forthcoming sections.
\\\\
Take the $R$-dependence of the Alexander polynomial \eqref{alexproperty0} as (one of) it's defining property. What could we then say about the group factors $r^R_{n,j}$? This defining property leads to nontrivial equations which group factors should satisfy. Our results are two properties of the solutions to these equations.
\paragraph{Description of group factors.}
It is complicated to extract independent group factors by direct calculation. For example in \cite{Labast} they are given to $6$'th order for fundamental representation of $SU(N)$ and for any spin representation of $SU(2)$. However, it is in fact a crucial problem, since finding independent group factors or at least enumerating them, means enumerating independent Vassiliev invariant. As discussed above, there is no dynamical data in the Wilson loop observables, only the representation theoretic and the topological. So in general, if one could list all the group factors for any Lie group, this would reduce knot invariants  to Vassiliev invariants, which is purely topological data.\\\\
The first result of the paper is a step in this direction, mainly we prove that \textbf{the independent solutions of \eqref{alexproperty0} are generated polynomials labeled by two-hook Young diagrams with two rows $\mathcal{Y}_2$:}
\begin{equation}\label{generators}
    X_{[y_1,y_2]}(R)= \dfrac{1}{y_1 y_2}C_{y_1}(R)C_{y_2}(R)+ \ldots , \qquad y_1 \geq y_2 \geq 2 , \quad [y_1,y_2] \in \mathcal{Y}_2
\end{equation}
where, the $C_k$ are $SU(N)$ Casimirs (see below). We do not have a closed form formula for these solutions, but they are constructed by a transformation from an explicit basis and specifically contain the term highlighted in \eqref{generators}.
The  \textbf{vector space basis $X_\lambda$ in the set of these solutions is labelled by Young diagrams with more then 1 hook -  $\mathcal{Y}$}. The way the basis $X_\lambda$ is constructed from $X_{[y_1,y_2]}$ follows from the main theorem of the paper.
\\\\
As a result we can write the following generating function:
\begin{equation}\label{AlexGen}
    \mathcal{X}(\hbar)= \sum_{\lambda \in \mathcal{Y}} X_{\lambda } v_\lambda(\hbar) 
\end{equation}
where $v_\lambda(\hbar)$ are generating functions of integers $v_{\lambda,n}$. 
There are two key statements to be made: 
\begin{itemize}
    \item We can observe experimentally, that the set of independent group factor of the Alexander polynomials is "smaller", then the described set of Young diagrams :
    \begin{equation}
        \Span \left\langle r^R_{n,j} \right\rangle 	\subsetneq \Span \left\langle  X_\lambda(R) \right\rangle
    \end{equation}
    To recover the Alexander polynommial for a given knot  from \eqref{AlexGen}  we should choose a specific point for $v_{\lambda,n}$, where some of then may vanish or be dependent:
    \begin{equation}\label{AlexfromX}
       \mathcal{A}_R^\K(q=e^\hbar) = \left.\mathcal{X}(\hbar)\right|_{v_{\lambda,n} = v_{\lambda,n}^\K}
    \end{equation}
    This also means that other yet unknown symmetry properties should exist even for the Alexander polynomial (in fact adding the rank-level duality discussed below, is not enough either.) We leave the question of finding such symmetries for future work.
    \item The important point is that cumbersome problem of describing the group factors via explicit evaluation of traces is actually described by  very simple combinatorics. In terms of chord diagrams/feynamn diagrams we could write \eqref{AlexGen} as:
    \begin{equation}
        \mathcal{X}=\sum_{D \in \mathcal{D}} X_{D } v_D(\hbar) 
    \end{equation}
    where now $\mathcal{D}$ would be the independent chord diagrams. The algebra of chord diagrams is complicated exactly due to the non-trivial relations among them. In fact even the dimension (the number of independent group factors) of this algebra is unknown, and only some asymptotic bounds exist.
\end{itemize}
 Our result means that there should be an explicit solution for the independent group factors expressed in simple terms of Young diagrams with more then 1 hook. Moreover, as said above, this set is even simpler as it is multiplicatively generated by \eqref{generators}.
\paragraph{Integrability.} 
Interestingly, we obtained the above results while examining the similarity between solutions $X_\lambda$ and and KP hierarchy equations that we observed in \cite{MMMMS}. A similarity between polynomials in Casimir eigenvalues and Hirota operators is not really satisfying, because Hirota operators do not form an algebra under multiplication. Upon further investigation we discovered, that the relation is precisely formulated in terms of soliton $\tau$-function.
Here we provide a proof that \textbf{the \emph{independent} solutions to property \eqref{alexproperty0} in each order $\hbar^n$ are in one-to-one correspondence with the dispersion relations of a one soliton $\tau$-function of the KP hierarchy.}\\ This is an explicit relation to a certain type of $\tau$ function. It is captured by the following equality:
\begin{equation}\label{corresp}
    \mathrm{KP}(z_i) \tau_{\text{soliton}}  \cdot  \tau_{\text{soliton}} = \mathcal{X}(\hbar)  \ \tau_{\text{soliton}}
\end{equation}
where by $ \mathrm{KP}(z_i)$ we denote is the generating function of \emph{all} KP equations, where Hirota equation appear as coefficients of $z_i$ expansion (see \eqref{kpgen} below) , and $ \tau_{\text{soliton}}$ is the soliton tau-function.  The identification will be explained below, while now we just mention, that the soliton momenta are identified with the Casimir eigenvalues:
\begin{equation}
    k_i=C_i(R)
\end{equation}
while the solution to the dispersion relations correspond to setting the representation to be a single-hook $R= [r,1^L]$.
\\\\
In fact to establish a correspondence to the Alexander polynomial, we should take the restriction \eqref{AlexfromX} in account. In terms of the KP hierarchy it should involve some restriction or modification of the soliton $\tau$-function. Anyway, since \eqref{AlexGen}  arises as a solution of certain properties of the $N=0$ limit of Wilson loops we would like to say, that by constructing \eqref{corresp} we explicitly demonstrate the emergence of KP integrability for Chern-Simons Wilson loops. 









\section{The Alexander polynomial}\label{sec:4}
Let $K$ be a knot in $S^3$ and $X_\infty$ be the infinite cyclic cover of the knot complement. The homology group $H_1(X_\infty)$ is a module over $\mathbb{Z}[q,q^{-1}]$, which is called the Alexander module. The Alexander polynomial $\A^\K(q)$ is usually defined as a generator of a certain ideal in this module \cite{Prasolov}.
One may notice that the Alexander polynomial is a special value of the fundamental HOMFLY polynomial:
\begin{equation}
\A^\K(q)={\cal H}^\K(q,1).
\end{equation}
HOMFLY polynomials can be defined for higher representations of the gauge group, therefore we have the following generalization: 
\begin{Definition}
The Colored Alexander polynomial is defined as a special value of the HOMFLY polynomial:
\begin{equation}\label{aldef}
\A^\K_R(q)={\cal H}^\K_R(q,1) \quad \text{or} \quad\ \A^\K_R(e^\h)=\lim_{N \rightarrow 0} {\cal H}^\K_R(e^\h,e^{N\h}).
\end{equation}
\end{Definition}
Representations of $SU(N)$ are labelled by Young diagrams. We identify a representation $R$ with it's diagram $[R_1,R_2,\ldots]$.
\begin{Definition}
A single hook diagram is a diagram of the form: 
\begin{equation}
R=[r,\underbrace{1,\ldots,1}_L]
\end{equation}
\begin{gather*}
\yng(5,3,2,1) \ \ \text{[5,3,2,1] diagram,}\quad \
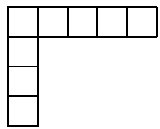
\yng(5,1,1,1) \ \  \text{Single hook diagram [5,1,1,1]}
\end{gather*}
\end{Definition}
\begin{Conjecture}
For any knot $\K$ and any single hook diagram $R=[r,1^L]$:
\begin{equation}\label{ac}
\A^\K_R(q)=\A^\K_{[1]}(q^{\vert R\vert}),
\end{equation}
where $|R|=r+L$ and $[1]$ is the fundamental representation.
\end{Conjecture}
This conjecture was proven for torus knots in \cite{Zhu}. It is also supported by explicit evaluations of HOMFLY polynomials colored by representations $[2,1]$ and $[3,1]$ for 3-strand knots \cite{MMMS21} and colored by representation $[2,1]$ for arborescent knots \cite{GuJ}. It is worth to note that this conjecture is a corollary of the eigenvalue conjecture, which drastically simplifies for Alexander polynomials colored by 1-hook diagrams. Based on this idea the sketch of the proof of \eqref{ac} is given in \cite{MM}.

Another important property is the symmetry of the HOMFLY (and, in particular, Alexander) polynomials with respect to the transposition of the Young diagram of the representation. This property holds for arbitrary diagrams $R$ and comes from the corresponding property of quantum groups
and WZW theories \cite{RankLevel1,RankLevel2}, in the latter case, it is called the rank-level duality (see also \cite{LiuPeng}):
\begin{equation}
{\cal H}_R^\K(q,a)={\cal H}_{R^T}^\K(q^{-1},a).
\end{equation}
This property is immediately inherited by the Alexander polynomials,
\begin{equation}\label{ranklevel}
\A^\K_R(q)=\A^\K_{R^T} (q^{-1}).
\end{equation}

\section{Alexander equations}\label{sec:5}
\subsection{The equations}
As we have mentioned the HOMFLY polynomial polynomial admits a series expansion in the formal variable $\h$:
\begin{equation}\label{Expansion1}
{\cal H}_R^\K=\sum_n \left(\sum_j v_{n,j}^\mathcal{K} r_{n,j}^R\right) \h^n.
\end{equation}
According to definition \eqref{aldef} the Alexander polynomial is a special value of the latter, therefore it inherits the structure of the expansion. Let us substitute \eqref{Expansion1} in  \eqref{a}, put $N=0$ and denote the resulting group factors as $A_{i,j}^R$
\begin{equation}
\begin{split}
\A^\K_R(q)-\A^\K_{[1]}(q^{\vert R\vert})=\sum_n \h^n &\sum_m {v}^\K_{n,m}
\left(r^R_{n,m}-|R|\cdot r^{[1]}_{n,m}\right)\Big|_{N=0}=:\\
&=: \sum_n \h^n \sum_m {v}^\K_{n,m} A^R_{n,m}\ \stackrel{^{R=[r,1^L]}}{=}\ 0.
\end{split}
\end{equation}

Since $\h$ is an arbitrary formal variable and this equality hold for all $\h$ we see that:
\begin{equation}\label{16}
\boxed{
A^{[r,1^L]}_{n,m}= 0.
}
\end{equation}
As described above, the group factors are expanded in the Lie algebra Casimir eigenvalues.
\begin{equation}\label{16p}
A^R_{n,m} = \sum_{|\Delta|\le n} \alpha_{\Delta,m} C_\Delta(R),
\end{equation}
where we label the monomials of $C_k$ by the Young diagrams in accordance with
\begin{equation}
C_\Delta=\prod_{i=1}^{l(\Delta)} C_{\Delta_i}.
\end{equation}
Then, $A^R_{n,m}$ can be considered as functions of Casimir invariants only, all the dependence on the representation entering through these latter ones:
\begin{equation}
A^R_{n,m}=A_{n,m}(C)
\end{equation}
Note that one can also reexpand the difference
\begin{equation}\label{17}
\begin{split}
\A^\K_R(q)-\A^\K_{[1]}(q^{\vert R\vert})&=\sum_n \h^n \sum_{|\Delta|\le n} C_\Delta(R) \sum_m {v}^\K_{n,m}  \alpha_{\Delta,m}= \\ =
&\sum_n \h^n \sum_{|\Delta|\le n} \alpha^\K_\Delta C_\Delta(R) \ , \text{where}\\
\alpha^\K_\Delta:=&\sum_m {v}^\K_{n,m}  \alpha_{\Delta,m}=(v_n^\K, \alpha_\Delta)\nonumber
\end{split}
\end{equation}
into the Casimir invariants instead of the group factors.

Equality \eqref{16} constitutes a property of the group factors of the Alexander polynomial. Let us look at generic polynomials of Casimir operators, that obey the desired property.
\begin{Definition}
The Alexander system of equations is a linear system of equations on the coefficients $\xi_{\Delta}^{(m)}$ given by 
\begin{equation}\label{gen}
\boxed{
X_{n,m} (C):= \sum_{|\Delta|=n} \xi_{\Delta}^{(m)} C_\Delta(R)=0}
\end{equation}
for any $n$ and any single hook diagram with $|R|=n$. 
\end{Definition}
It's obvious that the group-factors $A_{m,n}$ of the Alexander polynomial  are linear combination of the basis solutions to these equations:
\begin{equation}
A_{n,m}(C)\in \hbox{Span}\Big(\oplus_{k\le n}X_{k,m}(C)\Big)
\end{equation}
Therefore we consider (\ref{gen}) as equations defining the general structure of the polynomial. We denote by $\Al_n$ the vector space of solutions of the Alexander equations \eqref{gen} at order $n$, where by a solution we mean polynomials in Casimir invariants satisfying \eqref{gen}. We call a solution even, if it's an even polynomial in $C_k(R)$ and odd otherwise. The space of all solutions is the following graded vector space:
 \begin{equation}
Al=\bigoplus_n \Al_n
\end{equation}

The discussion about the trivalent graphs above allows to formulate the problem in other terms. We are studying Jacobi diagrams, which are weighted with an $\mathfrak{sl}_N$ weight system, with $N$ then set to zero. Hence the Alexander equations describe an ideal in the algebra of Jacobi diagrams with weights vanishing for single single hook representations of $\mathfrak{sl}_N$. We leave a more detailed analysis of this side of the problem for future studies.
\subsection{Properties of the Alexander equations}
Now let us explore some properties of the defined system \eqref{gen}.
\begin{Statement}
The Casimir invariants as functions of the Young diagram $R$ can be represented as the following shifted symmetric functions \cite{Zhe,OP,Casimir}:
\begin{equation}
\label{cas}
C_k(R)=\sum_{i=1}\Big[(R_i-i+1/2)^k-(-i+1/2)^k\Big].
\end{equation}
\end{Statement}
Restricted to 1-hook diagrams $R=[r,\underbrace{1,\dots,1}_L]$, this formula reduces to:
\begin{equation}\label{casrl}
C_k(R)=(r-1/2)^k-(-L-1/2)^k,
\end{equation}
or, in a more symmetric form, with $l=l(R)=L+1$ being the length of partition:
\begin{equation}\label{cas1}
C_k(R)=(r-1/2)^k+(-1)^{k+1}(l-1/2)^k.
\end{equation}
\begin{Corollary}
As a corollary of the explicit definition \eqref{cas1}, the symmetry with respect to transposition of the diagram reads
\begin{equation}\label{Transp}
C_\Delta(R^T)=(-1)^{|\Delta|+l(\Delta)}C_\Delta(R).\quad R=[r,1^L]
\end{equation}
\end{Corollary}
\begin{proof}
Simply notice, that the transposition of a single hook diagrams exchanges $r$ and $l$ introducing a factor of $(-1)^{\Delta_i+1} $ for the $i$'th multiplier.
\end{proof}
Generally one would think that solutions of the Alexander equations may contain even and odd polynomials. However, it appears that such solutions always factor into an even and odd part, which vanish separately.
\begin{Lemma}
The graded space $\Al$ of solutions of the Alexander equations decomposes into a direct sum of graded spaces of even and odd solutions:
\begin{equation}
\Al=\Al^e \oplus \Al^o
\end{equation} 
\end{Lemma}
\begin{proof}
Equation \eqref{gen} should hold for any 1-hook diagrams $R$, therefore it should be invariant under a transposition of the diagrams. Now suppose we have some solution $X_n(C)$ to \eqref{gen}, which means that for any 1-hook Young diagram $R$:
\begin{equation}
X_n=\sum_{|\Delta|=n} \beta_\Delta C_\Delta(R)=0
\end{equation}
where some $\beta_\Delta$ may vanish. Under transposition of $R$ each monomial behaves according to \eqref{Transp}. Since $|\Delta|$ is fixed after the transposition we get:
\begin{equation}
\sum_{|\Delta|=n} \beta_\Delta C_\Delta(R^T) = \sum_{|\Delta|=n} (-1)^{l(\Delta)}\beta_\Delta C_\Delta(R)=0
\end{equation}
Since $l(\Delta)$ is either odd or even, denote the even one's as $\Delta_1$ and the odd ones as $\Delta_2$, and split the sum:
\begin{gather*}
0 = \sum_{|\Delta_1|=n} \beta_{\Delta_1} C_{\Delta_1}(R) + \sum_{|\Delta_2|=n} \beta_{\Delta_2} C_{\Delta_2}(R) \\ = 
\sum_{|\Delta_1|=n} \beta_{\Delta_1} C_{\Delta_1}(R) - \sum_{|\Delta_2|=n} \beta_{\Delta_2} C_{\Delta_2}(R) \Rightarrow \\
\Rightarrow 
X_n^o=\sum_{|\Delta_2|=n} \beta_{\Delta_2} C_{\Delta_2}(R) =0 \quad , \quad X_n^e=\sum_{|\Delta_1|=n} \beta_{\Delta_1} C_{\Delta_1}(R) =0 \\   
\end{gather*}
This means that any solution $X_n \in \Al$ splits into the sum of $X_n^e \in \Al^e$ and $X_n^o \in \Al^o$ which vanish separately, therefore are solutions \eqref{gen}.
\end{proof}
Some examples of the solutions to the system are:
\begin{align}\label{alexsol}
&\h^4:X^e_{4,1}(C)=C_1^4-4C_1 C_3+3 C_2^2;  \nonumber\\[4pt]
&\h^5: X^e_{5,1}(C)=C_2 C_1^3-3 C_4 C_1+2 C_2 C_3, \nonumber\\
 &\h^5:X^o_{5,1}(C)=C_1 \left( C_1^{4}-4\,C_1C_3+3\,C_2^{2} \right) ;\nonumber\\[4pt]
 &\h^6: X^e_{6,1}(C)=4 C_1^2 \left(C_1^4-4 C_3 C_1+3 C_2^2\right),   \\
& \phantom{\h^6: } \ X^e_{6,2}(C)=2 C_3 C_1^3-3 C_2^2 C_1^2-8 C_3^2+9 C_2 C_4,\nonumber\\
    &\phantom{\h^6: } \  X^e_{6,3}(C)=C_3 C_1^3+3 C_2^2 C_1^2-9 C_5 C_1+5 C_3^2; 
 \nonumber \\
    &\phantom{\h^6: } \  X^o_{6,1}(C)=C_1^{2}C_4-2\,C_1C_2C_3+C_2^{3}
  \nonumber  \\
    &\phantom{\h^6: } \   X^o_{6,2}(C)=C_2\left( C_1^{4}-4\,C_1C_3+3\,C_2^{2} \right)\nonumber
\end{align}
\begin{Theorem}
The dimensions of the homogeneous components of $Al$ are given by: 
\begin{equation}\label{Dim}
\dim{\mathcal{A}l^e_n}=p_e(n)-\left\lfloor\dfrac{n}{2}\right\rfloor, \quad l(\Delta)\in 2\mathbb{Z}
\end{equation}
\begin{equation}
\dim{\mathcal{A}l^o_n}=p_o(n)-\left\lfloor\dfrac{n+1}{2}\right\rfloor, \quad l(\Delta)\in 2\mathbb{Z}+1,
\end{equation}
where $p_{e}(n)$ and $p_o(n)$ is the number of partitions of $n$ into even and odd number of integers respectively.
\end{Theorem}
\begin{proof}
Denote $x=r-1/2, y=-L-1/2$, then the Casimir invariants \eqref{cas1} take the form:
\begin{equation}\label{casxy}
C_k(R)=x^k-y^k
\end{equation}
Equation \eqref{gen} should hold for all $r,L$, therefore it should also hold for $x,y$. Therefore it is equivalent to the vanishing of the coefficients of powers of $x,y$. It may happen that not all of these coefficients are independent. Therefore our goal is to find the number of independent coefficients, which will give us the number of equations on the initial variables $\xi_\Delta$.
\\
First, let's treat the case of the even solutions. The number of variables in the case is just $p_e(n)$, the number of partitions of even length. Therefore monomials $C_\Delta$ always contain at least two Casimir operators. The product of two Casimir operators expands as:
\begin{equation}
C_{k_1}C_{k_2}=(x^{k_1}-y^{k_1})(x^{k_2}-y^{k_2})
\end{equation}
This means that the whole expression is divisible by $(x-y)^2$, therefore we are left with a polynomial in $x,y$ of order $n-2$. Now, notice that even though $C_k(R)$ is antisymmetric in $x,y$ the even solution is symmetric.\\\\
Hence, we have:
\begin{equation}
X_{n,m}(C)=(x-y)^2 \hat{X}_{n,m}
\end{equation}
$\hat{X}_{n,m}$ is also symmetric, therefore it belongs to the ring symmetric polynomials of order $n-2$ in two variables $\Lambda(x,y)$. Moreover, it lies in the homogeneous component $\Lambda_{n-2}(x,y)$. To find the independent equations on $\xi_\Delta$ we must a basis in this ring. One basis in the ring of symmetric functions is given by the complete symmetric polynomials:
\begin{equation}
h_k(x,y)=\sum_{\substack{l_1+l_2=k \\l_i\geq 0}} x^{l_1} y^{l_2}.
\end{equation}
The basis in $\Lambda_{n-2}(x,y)$ is given by the polynomials:
\begin{equation}
h_\lambda(x,y)=h_{\lambda_1}h_{\lambda_2}, \quad \lambda_1\geq\lambda_2\geq 0, \ \lambda_1+\lambda_2=n-2
\end{equation}
Now, the coefficients of each $h_\lambda(x,y)$ give rise to independent equations on $\xi_\Delta$.   
Hence, the dimension of the space in question is 
\begin{equation}
\dim{\Al_n}=p_e-\dim{\Lambda_{n-2}(x,y)},
\end{equation}
with
\begin{equation}
\dim{\Lambda_{n-2}(x,y)}=\left\lfloor\dfrac{n}{2}\right\rfloor
\end{equation}
 \\\\
The case of odd solutions is treated in the same manner. This time the shortest monomial is of length 1, therefore the equation is divisible by $(x-y)$. Dividing an expression, antisymmetric in $x,y$, by another antisymmetric expression we get once again a symmetric polynomial in $x,y$ this time of degree $n-1$. Carrying out a similar analysis once again, we obtain for the dimension of the space of odd solutions:
\begin{equation}
\dim{\Al_n}=p_o-\dim{\Lambda_{n-1}(x,y)},
\end{equation}
with
\begin{equation}
\dim{\Lambda_{n-1}(x,y)}=\left\lfloor\dfrac{n+1}{2}\right\rfloor
\end{equation} 
\end{proof}
\begin{Remark}
The generating function for $p_e(n)$ and $p_o(n)$ are: \begin{gather}
\dfrac{1+\vartheta_4(0,q)}{2 (q;q)_{\infty}}=\sum_{n=0}^\infty p_e(n) q^n \\
\dfrac{1-\vartheta_4(0,q)}{2 (q;q)_{\infty}}=\sum_{n=0}^\infty p_o(n) q^n \nonumber 
\end{gather}
\end{Remark}
\begin{Lemma}\label{3.2}
The group factors in the expansion of the Alexander equation lie in a certain subspace of $\Al$:
\begin{equation}\label{321}
A_{2 n,m}(C)\in \bigoplus_{k=2}^{n} \left( \Al^e_{2k}  \oplus \Al_{2k-1}^o \right) \subset 
\end{equation}
\begin{equation}\label{322}
A_{2 n+1,m}(C)\in \bigoplus_{k=2}^{n+1} \left( \Al^o_{2k}  \oplus \Al_{2k-1}^e \right) \subset \Al 
\end{equation}
In other words in every order $n$ of $\h$, the terms of degree $k= n \ {\rm mod} \ 2$  consist of even monomials $C_\Delta$, and the terms of degree $k= n+1 \ {\rm mod} \ 2$  consist of odd monomials.
\end{Lemma}
\begin{proof}
Let us apply the rank-level duality to the expansion components. Since the expansion corresponds to letting $q=e^{\h}$, for the group factor the property reads: 
\begin{equation}
A_{n,m}(C(R^T))=(-1)^n A_{n,m}(C(R)).
\end{equation}
At the same due to \eqref{Transp}, we acquire a factor of $(-1)^{|\Delta|+l(\Delta)}$ for each summand of order $|\Delta|=k$. For the two properties to be consistent we should have:
\begin{equation}
(-1)^n=(-1)^{k+l(\Delta)}.
\end{equation}
\end{proof}
Let us state another property, which motivates the later discussion:
\begin{Lemma}\label{oddeven}
The odd solutions can be generated by multiplying the polynomials from $\Al^e$ by suitable odd monomials $C_\Delta$ with $l(\Delta)$ - odd. 
\end{Lemma}
We will return to the proof of this lemma after proving the main theorem in the next sections. However it's useful to hold in mind, since it explains our concentration on the even solutions $Al_e$ in the forthcoming  section. For example:
\begin{eqnarray}
X^o_{5,1}(C) &=& C_1 X^e_{4,1}(C) \nonumber \\
X^o_{6,1}(C) &=& \frac{1}{3}C_2 X^e_{4,1}(C) - \frac{1}{3}C_1 X^e_{5,1}(C)  \\
X^o_{6,2}(C) &=& C_2 X^e_{4,1}(C) \nonumber
\end{eqnarray}
\subsection{Group factors of the Alexander polynomial}
Let us illustrate how these lemmas work for the group factors of the Alexander polynomial:
\begin{align*}
\A^\K_R(q)-\A^\K_{[1]}(q^{\vert R\vert})&=\hbar^4 {v}^\K_{4,1} A_{4,1}(C)+\\&+\hbar^6\Big({v}^\K_{6,1} A_{6,1}(C)+{v}^\K_{6,2} A_{6,2}(C)\Big)+\\&
+\hbar^7{v}^\K_{7,1} A_{7,1}(C)+\\&
+\hbar^8\Big(A_{8,1}(C) v^\K_{8,1}+A_{8,2}(C) v^\K_{8,2}+A_{8,3}(C) v^\K_{8,3}+A_{8,4}(C)v^\K_{8,4} \Big)+\\&+ O(\hbar^9)
\end{align*}
where $A_{n,m}$ are certain combinations of the basis solutions, that we have listed above.
\begin{align}
&A_{4,1}(C)=X^e_{4,1}(C) \nonumber\\
&A_{5,m}(C)=0\nonumber\\
&A_{6,1}(C)= X^e_{4,1}(C), \ \
A_{6,3}(C)=-X^e_{6,1}(C) - \frac{5}{3}X^e_{6,2}(C) - \frac{2}{3}X^e_{6,3}(C),\ \\
&A_{7,1}(C)=X^o_{6,1}(C) \nonumber\\
&A_{8,1}(C)=X^e_{4,1}(C), \ \   A_{8,2}(C)=A_{6,3}(C) \ \nonumber \\  
&A_{8,3}(C)=C_1^4X^e_{4,1}(C)+7 X^e_{8,6}(C)-7 X^e_{8,7}(C)+X^e_{8,8}(C) ,   \ A_{8,4}(C)= C_1^2 X_{4,1}^e(C)\nonumber
\end{align}
with the solutions of the Alexander equations at order 8 being:
\begin{equation}
\begin{split}
&X^e_{8,1}=-C_2^2 \left(C_1^4-4 C_3 C_1+3 C_2^2\right), \\
&
X^e_{8,2}=-\left(C_1^4-4 C_3 C_1+3 C_2^2\right) \left(3 C_2^2+4 C_1 C_3\right),  \\
&X^e_{8,3}=-C_2 \left(C_2 C_1^4-2 C_4
   C_1^2+C_2^3\right),\\
&
X^e_{8,4}=-4 C_3 C_1^5-7 C_2^2 C_1^4+16 C_5 C_1^3-5 C_2^4, \\
&X^e_{8,5}=5 C_1^4 \left(C_1^4-4 C_3 C_1+3 C_2^2\right), \\
&
X^e_{8,6}=\left(4 C_3 C_1^5-5 C_2^2 C_1^4+5 C_2^4+60 C_4^2-64 C_3 C_5\right),  \\
&X^e_{8,7}=\left(C_2^4-4 C_6 C_2+3 C_4^2\right), \\
&
X^e_{8,8}=4 C_3 C_1^5+11 C_2^2 C_1^4-64 C_7 C_1+21 C_2^4+28 C_4^2
\end{split}
\end{equation}
We notice that there are less independent group factors at each order, than there are solutions of the Alexander equations. Obviously, we have a restriction given by lemma \ref{3.2}. However, there are certainly more: for example the components $\Al^o_{2k-1}$ and $\Al^e_{2k-1}$ from \eqref{321},\eqref{322} are absent. Moreover, even from what is left not all of the solutions enter independently. 
This hints that there are more symmetries of the Alexander polynomial, that would exactly fix its group factors.

\subsection{Supersymmetric polynomials}
The formula \eqref{casxy} can be extended to generic Young diagrams by parametrize them in terms of hooks, i.e. by the diagonal boxes. In the same manner one can introduce the variables $x_i=r_i-1/2,\ y_i=-L_i-1/2$ for the $i$-th hook. Then the Casimir invariants read:
\begin{equation}
C_k(x_1,\ldots,y_1,\ldots)=\sum_{i=1}^n (x_i^k -y_i^k)
\end{equation}
This expression is symmetric in the shifted variables $x_i,y_i$ separately. Formally this means that the Casimir invariants are supersymmetric in the variables $x,y$.
\begin{Definition} A polynomial $P(x,y)$ is supersymmetric if it is invariant under separate permutations of $x$'s and $y$'s and for $x_i=y_i=z$ doesn't depend on $z$.
\end{Definition}
As in the usual case one can define the ring of supersymmetric polynomials and look for various basis in it. One of those is given by the supersymmetric generalization of power sums \cite{SS1}:
\begin{equation}
\mathfrak{P}(x_1,\ldots,y_1,\ldots)=\sum_{i=1}^n x_i^k - \sum_{j=1}^m y_j^k
\end{equation}
Therefore the Casimir invariants are special cases of these supersymmetric power sums. In these terms we can give another formulation of our problem:
$\Al_n$ is an ideal in the ring of superymmetric functions that vanishes when all but one pair of variables are set equal. The connection of  supersymmetric polynomials and the KP hierearchy was studied in \cite{SS2}.
\begin{Remark} The supersymmetric polynomials are usually discussed in the context of Lie superalgebras, namely the space generated by Casimir invariant of $gl_{n|m}$ isomorphic to the space of supersymmetric $(n,m)$ polynomials\cite{SS4}. 
\end{Remark}
\section{The KP hierarchy}\label{sec:6}
Now we need a few standard results about the KP hierarchy \cite{MiwaJimboDate,Kharchev,TG}.
\begin{Definition}
The operator
\begin{equation}
D^n_x: (f,g) \rightarrow (D^n_x f\cdot g)
\end{equation}
defined by \begin{equation}
f(x+y)g(x-y)=\sum_{n=0}^{\infty}\dfrac{1}{n!}(D^n_x f\cdot g)y^n
\end{equation}
is called a \emph{Hirota derivative}.
\end{Definition}
The multivariate Hirota derivative is defined in a similar fashion via the multivariate Taylor expansion:
\begin{equation}
f(\textbf{x}+\textbf{y})g(\textbf{x}-\textbf{y})=\sum_{i_1,\ldots, i_m}\left(D^{l_1}_1\ldots D^{l_m}_m f\cdot g\right) y_1^{l_1}y_2^{l_2}\ldots y_m^{l_m}
\end{equation}

Here we notice an important point. If $g=f$, then odd powers of Hirota derivatives act trivially for any $f$. 
\\The KP hierarchy is defined by the following bilinear identity:
\begin{equation}\label{bi}
\oint \dfrac{d \zeta}{2\pi i} e^{ 2 \sum _i \zeta^i z_i } e^{\sum _i -\frac{1}{i \zeta^i}D_{T_i}}e^{\sum _j z_j D_{T_j}} \tau \cdot \tau  = 0,
\end{equation}
where  $D_{T_i}$ are the Hirota derivatives, $\tilde D \equiv  (D_{T_{^1}}$, ${1\over 2} D_{T_{^2}},\ldots$), and $z=\{ z_1,z_2,\dots \}$ are formal parameters.
Consequently, it can be shown that  the KP hierarchy is defined by the following generating function:
\begin{equation}\label{hir}
\mathrm{KP}(z_i) \tau \cdot \tau=
\sum ^\infty _{i=0} \chi_i(-2z)\chi_{i+1}(\tilde D_T) e^{[\sum _jt_jD_{T_j}]} \tau \cdot \tau  = 0,
\end{equation}
where  $\chi_i$ are the Schur functions in symmetric representations. The KP equations in the Hirota form appear as coefficients of  (\ref{hir}) of shur function $\chi(z_i)$ in a polynomial form $P_n(D_1,D_2,\ldots,D_{n-1})$. The first few equations are: \\
$\bullet \quad$ The equation of order 4, which is the KP equation itself:
\begin{flushleft}
\begin{equation}
(4D_1 D_3 - 3D_2^2-D_1^4) \tau \cdot \tau = 0
\end{equation}
\end{flushleft}
$\bullet \quad$ Higher KP equations:
\begin{align}\label{Hirota5}
P_5 \  & =3D_1 D_4 - 2D_2 D_3 - D_1^3 D_2 \nonumber; \\[4pt]
P_{6,1}& =D_1^2 \left(D_1^4-4 D_3 D_1+3
   D_2^2\right),\nonumber \\
P_{6,2} & =D_1^6-20 D_3 D_1^3-45 D_2^2
   D_1^2+144 D_5 D_1-80 D_3^2,\nonumber \\
P_{6,3} & =D_1^6+10 D_3 D_1^3-36 D_5 D_1-20 D_3^2+45 D_2
   D_4;   \\[4pt]
\end{align}

The equations of the hierarchy follow from expanding the polynomial. It appears that some of these equations are not independent. Moreover some of the differ only by terms odd in powers of Hirota operators. Therefore we propose to look at a certain set of equations which do not contain operators that act trivially and also generate the KP hierarchy.
Denote the ring of even polynomials in variables $x_1,x_2,\ldots $ with integer coefficients by $P^e(x)$. It's naturally graded by the power of the polynomial $P^e(x)=\bigoplus_n P^e_n(x)$. 
\\\\

We want to establish a connection of the group factors with soliton $\tau$-functions. Let us remind how the soliton $\tau$-functions of the KP hierarchy are constructed. One considers a Hirota equation: 
\begin{equation}
P(D_1,D_2,\ldots)\tau \cdot \tau=0
\end{equation}
Look for a solution as a formal series:
\begin{equation}\label{solit}
\tau=1+\epsilon f_1 + \epsilon^2 f_2 +\ldots
\end{equation}
In the first order one gets the following dispersion relation for $f_1=e^{\sum_j k_j T_j}$:
\begin{equation}\label{64}
\dfrac{1}{2}(P(k_1,k_2,\ldots)+P(-k_1,-k_2,\ldots))=0
\end{equation}
Where the symmetrization appears as a consequence of the trivial action of odd Hirota derivatives. The action of the KP hierarchy on the soliton $\tau$-function produces a set of dispersion relations. As these are just polynomials in the momentum variables we are free to generate an ideal in the ring of even polynomial by these dispersion relations. Denote the ring of even polynomials in variables $x_1,x_2,\ldots $ with integer coefficients by $P^e(x)$. It's naturally graded by the power of the polynomial $P^e(x)=\bigoplus_n P^e_n(x)$. 
\begin{Definition}
Define the Hirota ideal $\widehat{\KP} \subset P^e(k)$ as the ring of polynomial generated by the dispersion relations of the KP hierarchy
\begin{equation}
\widehat{\KP} = \Span \langle \ 
\text{Dispersion relations of the soliton $\tau$-function} \ 
 \rangle
\end{equation}
Different order KP equations produce different order dispersion relation, i.e. the Hirota ideal is  naturally graded: $\widehat{\KP}=\bigoplus_n \widehat{\KP}_n$
As it was mentioned, \eqref{bi} generates all Hirota equations, hence it also generates the dispersion relations when $\tau$ is specified to be of the form \eqref{solit}. Therefore, the generating function of the dispersion relations looks like:
\begin{equation}
\oint \dfrac{d \zeta}{2\pi i} e^{ 2 \sum _i \zeta^i z_i } e^{\sum _i -\frac{1}{i \zeta^i}k_i}e^{\sum _i z_i k_i}  +
\oint \dfrac{d \zeta}{2\pi i} e^{ 2 \sum _i \zeta^i z_i } e^{ \sum _i \frac{1}{i \zeta^i}k_i}e^{- \sum _i z_i k_i} .
\end{equation}
While having a generating function, there is no convenient way to enumerate all of the dispersion relations in such way that their linear independence would be evident. However we can find a subset, which has this property.
\end{Definition}
For $Y=[i,j], |Y|=i+j, i\geq j\geq 2$ define the polynomials:
\begin{equation}\label{kpgen}
h_Y(k)  \ =  \ \begin{vmatrix}
\chi_{i} (-\tilde{k}/2) & \chi_{i} (\tilde{k}/2) & \chi_{i+1} (\tilde{k}/2) \\
\chi_{j-1} (-\tilde{k}/2) & \chi_{j-1} (\tilde{k}/2) & \chi_{j} (\tilde{k}/2) \\
1 & 1 & \dfrac{1}{2}k_1
\end{vmatrix} 
+ 
\left(k \rightarrow -k \right).
\end{equation}
where $\tilde{k}=\left\{\dfrac{k_i}{i}\right\}$
\begin{Proposition}
The polynomials \eqref{kpgen} are some of the dispersion relations of the KP hierarchy and are linearly independent.
\end{Proposition}
\begin{proof}
By \cite{Natanzon} the equations:
\begin{equation}\label{n1}
\begin{vmatrix}
\chi_{i} (-\tilde{D}/2) & \chi_{i} (\tilde{D}/2) & \chi_{i+1} (\tilde{D}/2) \\
\chi_{j-1} (-\tilde{D}/2) & \chi_{j-1} (\tilde{D}/2) & \chi_{j} (\tilde{D}/2) \\
1 & 1 & \dfrac{1}{2}D_1
\end{vmatrix} \tau \cdot \tau =0 
\end{equation}
are a subset of the equations of the hierachy. Hence as in \eqref{64} the polynomials \eqref{kpgen} form dispersion relations. To see their linear independence one expands the first terms of Schur polynomials and computes the determinant. It appears that:
\begin{equation}
h_Y(k)=\dfrac{k_i k_j}{ij}-\dfrac{k_{i+1}k_{j-1}}{(i+1)(j-1)}+
\sum_{q=2}^{\infty}
 \sum_{r_1+\ldots+r_{2q}=|Y|}
 P^{ij}_{r_1 \ldots r_{2q}} k_{r_1} \ldots k_{r_{2q}}
\end{equation}
These polynomials are uniquely defined by their second order terms and are linearly independent.
\end{proof}
Let us denote the ideal generated by \eqref{kpgen} by $\KP=\bigoplus_n\KP_n$. Then we have the obvious inclusion:
\begin{equation}
\KP \subset \widehat{\KP}
\end{equation}
\section{KP hierarchy and solutions to Alexander equation.}\label{sec:7}
This section is be devoted to the announced proof of our statements given in section 3, in particular the refinement of the result of \cite{MMMMS} . It consists of two parts, which are formulated as separate lemmas: the first basically repeats a simple calculation in \cite{MMMMS} now interpreted correctly in terms of dispersion relations. The second lemma is a combinatoric construction, which completes the proof of the KP-Alexander relation and also provides the independent polynomials \eqref{generators}. There we also explain how the linear basis is given by diagrams $\mathcal{Y}_e$ and explain how $X_\lambda ,\, \lambda \in \mathcal{Y}_e $ is multiplicatively generated by $X_[y_1,y_2]$.
\begin{Theorem}
The Hirota ideal and the even part of the ring of solutions of the Alexander equation are in fact the same graded ideal, expressed in different variables.
\end{Theorem}
\begin{center}
\begin{Remark}
The single hook property of the Alexander polynomial is equivalent to the system of dispersion equations for the KP one soliton $\tau$-function.
\end{Remark}
\begin{Remark}
As mentioned in lemma \ref{oddeven} the odd part of the ring of solutions of the Alexander equations is generated by the even part (see Corollary \ref{oddeven2}) 
\end{Remark}

\end{center}
\begin{proof}
\begin{Lemma}\label{main1}
Any generator of $\widehat{\KP}$ satisfies the of the Alexander equation. In other words the Hirota ideal is an ideal in the ring of even solutions of the Alexander system:
\begin{equation}
\fullKP \subset \Al^e
\end{equation}
\end{Lemma}
\begin{innerproof}[Proof of Lemma \ref{main1}]
The Hirota equations are fully specified the bilinear identity (recall, that it gives rise to the generating function). To prove that every polynomial in the Hirota ideal is an even solution of the Alexander equations substitute $D_k$ with $C_k(R)$ in the symmetrized bilinear identity \eqref{bi}. The resulting integral:
\begin{equation}\label{a}
\oint \dfrac{d \zeta}{2\pi i} e^{ 2 \sum _i \zeta^i z_i } e^{\sum _i -\frac{1}{i \zeta^i}C_i}e^{\sum _i z_i C_i}  +
\oint \dfrac{d \zeta}{2\pi i} e^{ 2 \sum _i \zeta^i z_i } e^{ \sum _i \frac{1}{i \zeta^i}C_i}e^{- \sum _i z_i C_i} .
\end{equation}
vanishes when $R$ is a 1-hook diagram. Indeed, the value of the Casimir invariant on 1-hook diagrams is given by \eqref{casxy}. Simplify one of the exponential in the integrand:
\begin{equation}
e^{-\sum _i \frac{1}{i \zeta^i}C_i}=e^{-\sum _i \frac{1}{i \zeta^i}(x)^i}e^{\sum _i \frac{1}{i \eta^i}(y)^i}=\dfrac{\zeta-x}{\zeta-y},
\end{equation}
while the same exponential in the second integral is the reciprocal of this expression.\\
This allows us to compute the integrals taking the residues at $\zeta_1=x, \ \zeta_2=y$. Therefore, restricted to 1-hook diagrams, \eqref{a} vanishes:
\begin{equation}
(y-x)e^{-\sum z_i(x^i+y^i)}+(x-y)e^{-\sum z_i (x^i+y^i)}=0.
\end{equation}
We see that any polynomial arising from the bilinear identity vanishes, when Hirota derivatves are replaced with 1-hook values of Casimir invariants \eqref{cas}. This proves the lemma.\\\\\\
One may notice that this calculation resembles the general proof of the bilinear identity \cite{MiwaJimboDate}. What basically happens here is that $k_j=C_j(R)$ for any single hook diagram $R$ solves the dispersion relation. Therefore the values of the Casimir operators on single hook diagrams are the momentum variables in the soliton $\tau$-function of the KP hierarchy. To obtain an $n$-soliton $\tau$-function we need to have $n$ different solutions to the dispersion relation, which in the terms of Casimir invariants would just mean taking $n$ different hook diagrams:
\begin{equation}\label{tau}
\tau(t_1,t_2,\ldots)=\sum_{H'\subset H_n} c_{H'} \exp\left(\sum_{R\in H'}\xi(R)\right)
\end{equation}
where $H$ is a set of exactly $n$ single hook diagrams, $c_{H'}$ is a constant, the sum runs over all subsets of $H$ and
\begin{equation}
\xi(R)=\sum_{k=1}^{\infty}C_k(R)t_k
\end{equation}
\begin{Remark}
In these terms the reduction to the KdV $\tau$-function look quite natural, mainly it means taking hook diagrams that are symmetric with respect to transposition.
\end{Remark}
\end{innerproof}
Now since we have that every polynomial (dispersion relation) from the Hirota ideal satisfies the Alexander system of equations we will prove by comparing dimensions that they are equal. Lemma \eqref{main1} obiously guarantess, that each for homogenous components $\fullKP_n \subset \Al^e $.
\begin{Lemma}\label{main2}
The vectors space dimension of the homogeneous components of $\KP$ are given by \eqref{Dim}.
\end{Lemma}
Hence by this lemma we will have the following chain of inclusions:
\begin{center}
\begin{tikzcd}[row sep=tiny]
& \fullKP \arrow[dd,hook'] \\
\KP \arrow[ur,hook] \arrow[dr,equal] & \\
& \Al^e
\end{tikzcd}
\end{center}
We see that all of the inclusions are in fact equalities. 
\begin{innerproof}[Proof of Lemma \ref{main2}]
First, let us interpret the quantities \eqref{Dim}. Let us call the Young diagrams having $n$ diagonal boxes $n$ hook diagrams in concordance with one hook diagrams. Then \eqref{Dim} is the number of Young diagrams with more than 1 hook and of even length. We will denote this set of diagrams as $\mathcal{Y}_e$  \\
Now, as one can notice the polynomials, that generate the Hirota ideal \eqref{kpgen}, are naturally numbered by 2-hook diagrams with two rows $\mathcal{Y}_2$, which is nice since it's a subset of $\mathcal{Y}_e$.
\\\\
\begin{tikzpicture}[>=triangle 45,font=\sffamily]
    \node (H) at (0,0) {Solutions of Alexander equation};
    \node (Y) [right=1cm of H]  {\begin{minipage}{5cm}Diagrams with more than 1 hook:\begin{equation*}    
    \ytableausetup{smalltableaux}\ytableausetup{aligntableaux = center} \ydiagram{5,4,3,3,2}\end{equation*}\end{minipage}};
    \node (Z) [below=2cm of Y] {\begin{minipage}{5cm}Young diagrams with 2 rows and 2 hooks: \begin{equation*}\ytableausetup{smalltableaux}\ytableausetup{aligntableaux = center} \ydiagram{5,2}\end{equation*}\end{minipage}};
    \node (X) [left=3cm of Z] {$\KP_n$};
   
    \draw [semithick,->] (H) -- (Y) ;
    \draw [semithick,<->] (Y) -- (Z) node [midway,right] {$\mathcal{Y}_2 \subset \mathcal{Y}_e$};
    \draw [semithick,<->] (Z) -- (X);

   ;
\end{tikzpicture}

First of all let us perform a triangular transformation in $\KP$ , to get for $Y=[y_1,y_2]$:
\begin{gather}
h_Y(k)=\dfrac{1}{y_1 y_2}k_{y_1} k_{y_2} - \dfrac{1}{|Y|-1} k_{|Y|-1}k_1+ \ldots \  
\end{gather}
We abuse that notation and use $h$ for these new polynomials. Now the first term reflects the Young diagram associated to the polynomial, while the second only depends on the order. From now on we are not interested in carrying the coefficients since they exactly correspond to the monomial. We will symbolically denote the polynomials as follows:
\begin{equation}
h_Y(k)=[y_1,y_2]-\left[|Y|-1,1\right]
\end{equation} 
Any element from $\Al^e_n$ is represented by a Young diagram. However, contrary to the $\KP_n$ case, there is no direct correspondence yet, i.e. we can't tell the structure of a polynomial, which is represented by a diagram $T$. Take some element from $\Al^e_n$ and the corresponding diagram $T$ and cut it into two halves: the first two rows and everything else.
\begin{equation}
T=[Y,Y'],\quad T\in   \Y_e, \ Y\in\Y_2
\end{equation}
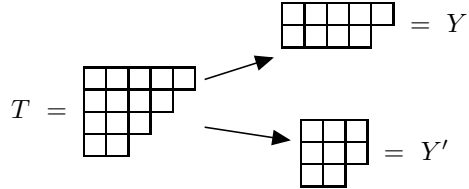
\begin{center}
\begin{tikzpicture}[>=triangle 45,font=\sffamily]
    \node (1) at (0,0) {$\ytableausetup{smalltableaux}
\ytableausetup{aligntableaux = center}
T \ = \ \ydiagram{5,4,3,2}$};
    \node (2) [above right=0cm and 0.9cm of 1]  {$\ytableausetup{smalltableaux}
\ytableausetup{aligntableaux = center}\ydiagram{5,4} \ = \ Y$};
    \node (3) [below=0.7cm of 2] {$\ytableausetup{smalltableaux}
\ytableausetup{aligntableaux = center}\ydiagram{3,3,2} \ = \ Y'$};
    
    \draw [semithick,->] (1) -- (2); 
    \draw [semithick,->] (1) -- (3);
\end{tikzpicture}
\end{center}

Since $T$ is at least a 2-hook diagram $Y$ is also a 2-hook diagram. Therefore to $Y$ we can assign a polynomial $h_Y(k)$ and to $Y'$ a monomial $k_{Y'}$. Vice versa, each $T$ can be constructed by putting together two partitions $Y$ and $Y'$ subject to necessary conditions. That way we can assign a polynomial to the resulting diagrams $T$:
\begin{equation}
\hat{h}_T(k)=h_Y(k)k_{Y'}
\end{equation}
\\\\
Let us illustrate the procedure. Take $Y$ - a partition of $n-2$, say $Y=[n-4,2]$, then the only possibility to form a partition of $n$ of even length is to add a $Y'=[1,1]$.
This way we construct:
\begin{equation}
\hat{h}_{[n-4,2,1,1]}=h_{[n-4,2]}k_1^2
\end{equation}
\\\\%
We want to inductively build the polynomials corresponding to the diagrams following the steps above. Our concern is the linear independence of the resulting polynomials, since our initial goal was counting the dimensions of the Hirota ideal.
Since we know the generating set $KP$, start with taking a 2 row 2 hook partition $Y$. The first non-trivial case is $|Y|=n-2$, which we have already described. The case $|Y|=n-3$ is treated exactly in the same manner, however, let us right it down for clearance. The partitions $Y$ have the form:
\begin{gather*}
[n-5,2]-[n-4,1]\\
\vdots \\
[n-5-i,2+i]-[n-4,1]\\
\end{gather*}
We can only multiply by a monomial $k_2 k_1$, denoted by $[2,1]$. Therefore, by combining these partitions we get the general form:
\begin{equation}\label{hh}
\hat{h}_{[n-5-i,2+i,2,1]}=[n-5-i,2+i,2,1]-[n-4,2,1,1]
\end{equation}
Now on the next step additional actions are needed. A general polynomial in $\mathcal{Y}_e$ for $|Y|=n-4$ will have a form:
\begin{equation}
[n-6-i,2+i]-[n-5,1]
\end{equation} Here we have two diagrams, that we can attach to the bottom:
\begin{equation}
Y'=[3,1] \leftrightarrow k_3 k_1 ,\qquad Y'=[1,1,1,1] \leftrightarrow k_1^4.
\end{equation}
While following our construction we encounter a polynomial:
\begin{equation}
[n-6,3,2,1]-[n-5,3,1,1],
\end{equation}
which we get by attaching $[3,1]$ to $[n-6,2]-[n-5,1]$. We notice that the first summand has actually already appeared in \eqref{hh}, for $i=1$. Even though the second terms are different this hints us that the polynomials built on this step might not be linearly independent with the ones constructed in the first step.\\\\
We want to construct a set of polynomial which are \textbf{explicitly independent}.
\\\\
To do this, notice that we can find a polynomial $\hh_T$ for $|Y|=n-2$, whose first term will be $[n-4,2,1,1]$ and one with $[n-5,3,1,1]$. These are the polynomials $\hh_{[n-4,2,1,1]}$ and $\hh_{[n-5,3,1,1]}$. This leads to a linear transformation and a redefinition:\\\\For $|Y|=n-3$:
\begin{equation}
\begin{split}
\hh_{[n-5-i,2+i,2,1]} \longrightarrow &\hh_{[n-5-i,2+i,2,1]}+\hh_{[n-4,2,1,1]}=\\&=[n-5-i,2+i,2,1]-[n-3,1,1,1]
\end{split}
\end{equation}
For $|Y|=n-4$:
\begin{equation}
\begin{split}
\hh_{[n-6-i,2+i,2,1]} \longrightarrow &\hh_{[n-6-i,2+i,2,1]}+\hh_{[n-5,3,1,1]}=\\&=[n-6-i,2+i,2,1]-[n-3,1,1,1]
\end{split}
\end{equation}
This transformation makes the constructed polynomials explicitly independent since the all have different monomials as their first term.
\\\\
Let us prove the general statement.
\begin{Proposition}
For all $T$ - $k$-hook Young diagrams,such that $|T|=n$, with $l(T)$ - even and $k>1$, there exists such transformation of $\hh_T$ such that:
\begin{equation}
\hat{h}_T= h_Y k_{Y'} \ , \quad T=[Y,Y']
\end{equation}
are linear independent.
\end{Proposition}
\begin{proof}
Along the lines of the given examples we want to prove that for any given $Y \vdash n-k$ and $Y'\vdash k$ such that $T=[Y,Y']$ is a partition too, the corresponding polynomial $\hat{h}_T=h_Y k_{Y'}$ is defined only by the first summand, which is precisely $k_T$.\\
Let us start by noticing that in case, when $Y'$ have different lengths, the resulting polynomials will for sure be linearly independent. Therefore we can consider the cases of $l(T)=4,6,8,\ldots$ separately.\\
1) Let's first construct the case $l(T)=4 \Leftrightarrow l(Y')=2$. Denote $Y'=[b_1,b_2], b_1 \geq b_2 \geq 1$, then 
\begin{equation}
\hat{h}_T=[n-k-2-i,2+i,b_1,b_2]-[n-k-1,b_1,b_2,1] \ , \quad  \text{if } 2+i\geq b_1
\end{equation}
We have demonstrated the cases $k=2,k=3$ before. Now suppose all polynomials for $k<k_0$ are brought in the form:
\begin{equation}
\hat{h}_T=[T]-[n-3,1,1,1]
\end{equation}
Then to do the same for $k=k_0$ we must find the term $Y=[n-k_0-1,b_1,b_2,1]$ in the already constructed polynomials simply by solving a simple equation 
\begin{equation}\label{simple}
 n-k-2-i=n-k_0-1 \ , \ b_1=2+i.
 \end{equation}
Therefore we find the needed term for $i=b_1-2 \geq 0$ and $k=k_0-b_1+1$, with $k_0>k\geq 2$ since $b_1 \leq k_0-1$. Finally we set the corresponding $b'_1=b_2,b'_2=1$ .That way by induction we get all the polynomials in the desired form and their independence becomes explicit. The inequalities guarantees the existence of such term.\\ \\
Let us illustrate the last step by another example.
Take, for instance, $k_0=5$ and $Y'=[3,2]$ and assume for $k<5$ everything is sorted out. We begin with:
\begin{equation}
\hat{h}_T=[T]-[n-6,3,2,1]
\end{equation}
We have $i=1$,$k=3, b_1'=2, b_2'=1$. This means we have to look for a polynomial of the form:
\begin{equation}
\hat{h}_T=[n-3-2-1,2+1,2,1]-[n-3,1,1,1]=[n-6,3,2,1]-[n-3,1,1,1],
\end{equation}
which we should have built in the previous steps of the induction. 
\\\\
2) The cases $l(Y') > 2$ are treated in the same way. Mainly for each such case $l(Y')=l$ we will have a minimal $k=l$, where such $|Y'|=k$ is possible. The corresponding monomial will be $k_{Y'}=k_1^{l}$.\\
A general polynomial will look like:
\begin{equation}
h_T=[n-k-2-i,2+i,b_1,\ldots b_l]-[n-k-1,b_1,\ldots,b_l,1]
\end{equation}
We want to bring them all in the form:
\begin{equation}
\hat{h}_T=[n-k-2-i,2+i,b_1,\ldots b_l]-[n-k'-1,1^l,1]=[T]-[n-k-1,1^l,1]
\end{equation}
Again at each level $k_0$ we want to cancel the second term by a first term of some polynomial with $k<k_0$. Then by induction we can bring all the polynomials in the desired form. Let us see how this is realized.\\
At level $k_0$ we want to cancel the term:
\begin{equation}
[n-k_0-1,b_1,\ldots,b_l,1] \text{ by } [n-k-2-i,2+i,b'_1,\ldots b'_l]
\end{equation}
We see that we need to solve the same equations ~\eqref{simple} to see that such cancellation is indeed possible. 
\end{proof}

\end{innerproof}
\end{proof}
As announced lemma \eqref{oddeven} is a corollary of the theorem.
\begin{Corollary}\label{oddeven2}
The space of odd solutions $\Al^o_n$ is generated by even solutions of lower orders.
\end{Corollary}
\begin{proof}
We have naturally labelled all the even solutions of the Alexander equations by Young diagrams of even length having more than 1 hook - $\Y_e$. A solution, labelled by $T\in\Y_e$ contains a unique term $C_T$. Now, let's interpret the quantity \eqref{Dim} for the odd case. Again it is exactly the number of diagrams of odd length and having more than one hook - $\Y_o$.\\
This means we may use the same idea. Take any diagram $T\in\Y_o$, and cut the lowest row $T_{l(T)}$, this way we get a diagram $T'$ from $\Y_e$. This means $h_{T'}(C)$ is a even solution of the Alexander equations. Hence, the product:
\begin{equation}
h_T(C)=h_{T'}(C)C_{T_{l(T)}}
\end{equation}
is a odd solution to the Alexander system. Of course all this may be done with any diagram from $\Y_o$. Moreover, these odd solutions are also independent since they have a unique term, corresponding to the diagram, they are labelled by.

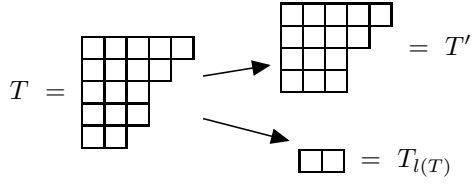
\begin{center}
\begin{tikzpicture}[>=triangle 45,font=\sffamily]
    \node (1) at (0,0) {$\ytableausetup{smalltableaux}
\ytableausetup{aligntableaux = center}
T \ = \ \ydiagram{5,4,3,3,2}$};
    \node (2) [above right=-1cm and 0.9cm of 1]  {$\ytableausetup{smalltableaux}
\ytableausetup{aligntableaux = center}\ydiagram{5,4,3,3} \ = \ T'$};
    \node (3) [below=0.5cm of 2] {$\ytableausetup{smalltableaux}
\ytableausetup{aligntableaux = center}\ydiagram{2} \ = \ T_{l(T)}$};
    
    \draw [semithick,->] (1) -- (2); 
    \draw [semithick,->] (1) -- (3);
\end{tikzpicture}
\end{center}
\end{proof}
The fact that we only need to study the even part of the solutions is in tact with the triviality of the odd Hirota equations. Hence really there is a correspondence between the whole space $\Al$ and the dispersion relations of the soliton $\tau$-functions.
Another obvious, however important corollary is the statement announced above.
\begin{Corollary}\label{generatingcorollary}
The space of all solutions to the Alexander equations is generated by a set labelled by two row, two hook diagrams $\mathcal{Y}_2$:
\begin{equation}
    X_{[y_1,y_2]}(R)= \dfrac{1}{y_1 y_2}C_{y_1}(R)C_{y_2}(R)+ \ldots , \qquad y_1 \geq y_2 \geq 2
\end{equation}
while the linearly independent solutions are labeled by diagrams with more then one hook $\mathcal{Y}=\mathcal{Y}_e\bigoplus \mathcal{Y}_o$:
\begin{equation}
    X_\lambda = X_{[\lambda_1,\lambda_2]}  C_{\lambda'}
\end{equation}
where $\lambda'= [\lambda_3, \ldots , \lambda_{l(\lambda)}]$. 
\end{Corollary}

\section{Discussion}
In this paper we have refined the statement made in paper \cite{MMMMS} and then prove it. Our theorem connects the solutions of the Alexander equations with the dispersion relations of 1-soliton tau-function. The special form of the group factor of the Alexander equations is the dispersion relation of the KP soliton. In these terms the property of the polynomial holds because the special form of the Casimir invariants is exactly the one needed to solve these relations. From our point of view, this idea brings a clearer understanding to the initially mysterious coincidence between the Alexander polynomial group factors and the KP Hirota equations/$\tau$-functions.
\\
Our results open various new perspectives:
\begin{itemize}
\item Does the property of the colored Alexander polynomial generalize in some way to diagrams with two and more hooks? If it does it should reflect on the KP side of the story and modify the construction.

\item Let us consider an n-soliton $\tau$-function. This would mean \eqref{tau} finding $n$ different solutions to the dispersion relation, i.e. considering several different values of the Casimir invariants. This might be a way to invariants of links.

\item As mentioned not all of the solutions of the Alexander equations enter the Alexander polynomial itself. Therefore there is a set of symmetry conditions \cite{MST}, which sort out some of the solutions (See section 3.3). It is interesting to find out how this new symmetry changes the integrable property.

\item The supersymmetric nature of these solutions may also open an interesting way of looking at the Alexander polynomial.

\item We found that the two limits of the HOMFLY polynomials are connected to two types of $\tau$-functions of the KP hierarchy and, hence, we are closer to understating the integrable properties of the HOMFLY polynomial for generic variables.

\item Completely different approach to study various properties of integrable tau-functions is given by topological recursion \cite{AMM/CEO}. In the context of polynomial knot invariants it is developed only for torus knots \cite{BEM,DBPSS}. However, even in this special case, it would be interesting to reproduce our results using topological recursion methods. 

\item Colored HOMFLY polynomials can be described in terms of quantum R-matrices. The dependence of representation is fully encoded by quantum 6j-symbols, hence our results are the result of some symmetries of these 6j-symbols \cite{symm}, probably, related to the eigenvalue conjecture \cite{evc}. 

\end{itemize} 

\section*{Acknowledgments}
This work was funded by the Russian Science Foundation (Grant No.16-12-10344).  We are grateful to A.Morozov and A.Mironov for stimulating discussions and helpful insights. A.S. is especially thankful to A.Zabrodin, S.Natanzon and A.Marshakov for explanations and advice. V.M. is indebted to P. Dunin-Barkowski for carefully reviewing and commenting the paper.

\end{document}